\documentclass{IEEEtran}
\usepackage{float}
\usepackage[caption = false]{subfig}
\usepackage{siunitx}
\usepackage{setspace} 
\usepackage{cite}
\usepackage{amsmath,amssymb,amsfonts}
\usepackage{amsthm}
\theoremstyle{remark}

\theoremstyle{remark}
\newtheorem{remark}{Remark}
\theoremstyle{remark}

\theoremstyle{remark}
\newtheorem{definition}{Definition}
\theoremstyle{remark}

\theoremstyle{remark}
\newtheorem{lemma}{Lemma}
\theoremstyle{remark}

\usepackage{graphicx}
\usepackage{hyperref}
\usepackage{enumerate}
\usepackage{textcomp}
\usepackage{xcolor}
\usepackage{soul}
\usepackage{algorithm}
\usepackage{algpseudocode}
\usepackage{booktabs}
\usepackage{authblk}
\usepackage{fancyhdr}

\begin{document}
\title{Data-Driven Inverse Optimal Control for Continuous-Time Nonlinear Systems
}

\author{Hamed Jabbari Asl$^{\ast}$ and Eiji Uchibe
\thanks{$^{\ast}$The authors are with the Department of Brain Robot Interface, ATR Computational Neuroscience Laboratories, 2-2-2 Hikaridai, Seikacho, Soraku-gun, Kyoto 619-0288, Japan (e-mails: \texttt{\{hjabbari, uchibe\}@atr.jp}).}
\thanks{This work has been submitted to the IEEE for possible publication.
Copyright may be transferred without notice, after which this version may
no longer be accessible.}
}
\date{}

\maketitle
%********************************************
%********************************************
\begin{abstract}
This paper introduces a novel model-free and a partially model-free algorithm for inverse optimal control (IOC), also known as inverse reinforcement learning (IRL), aimed at estimating the cost function of continuous-time nonlinear deterministic systems. Using the input-state trajectories of an expert agent, the proposed algorithms separately utilize control policy information and the Hamilton-Jacobi-Bellman equation to estimate different sets of cost function parameters. This approach allows the algorithms to achieve broader applicability while maintaining a model-free framework. Also, the model-free algorithm reduces complexity compared to existing methods, as it requires solving a forward optimal control problem only once during initialization. Furthermore, in our partially model-free algorithm, this step can be bypassed entirely for systems with known input dynamics. Simulation results demonstrate the effectiveness and efficiency of our algorithms, highlighting their potential for real-world deployment in autonomous systems and robotics.

\end{abstract}
%********************************************
%***************************************
\section{Introduction}
Inverse optimal control (IOC), often referred to as inverse reinforcement learning (IRL) when approached through data-driven methods, is a powerful framework
used in the field of control theory and machine learning
to infer the underlying cost functions or objectives that govern the observed behaviors of a system or agent. Identifying these objectives is critical across various domains, including autonomous robotics \cite{uchibe2020imitation}, human-robot interaction \cite{kollmitz2020learning}, and multi-agent systems \cite{golmisheh2024optimal}. %,  and economics \cite{sanghvi2021inverse}. 
By uncovering the implicit goals driving an agent's behavior, we can enhance decision-making processes, optimize system performance, and enable robots to navigate and operate more effectively in complex, dynamic environments.

In practical scenarios, designing an objective function for an autonomous agent often involves iterative optimization processes, commonly referred to as reward engineering. This process is labor-intensive and may fail to capture the nuanced behaviors exhibited by expert agents. Leveraging the estimated objective functions of expert agents, therefore, becomes crucial for adapting objectives in other autonomous systems \cite{silver2021reward}.

In the realm of computer science, IRL techniques predominantly rely on substantial training data to estimate reward functions in stochastic environments \cite{abbeel2004apprenticeship, abbeel2005exploration, ratliff2006maximum}. However, IOC/IRL methods developed in the control systems community often operate with limited data, typically estimating the cost function from a single demonstration in deterministic systems \cite{kamalapurkar2018linear}.

IOC/IRL methods for deterministic systems, which primarily address the estimation of cost functions that penalize both states and inputs, vary in their methodologies and assumptions. Some rely solely on data from expert demonstrations, circumventing the need for direct interaction between the learner and the environment. These methods, whether model-based \cite{kamalapurkar2018linear, self2019online, SELF2022110242} or partially model-based \cite{Asl_IRL22}, require knowledge of the system's input dynamics. Consequently, challenges arise in effectively estimating input dynamics when they are not known in advance. Alternatively, model-free approaches have emerged, requiring direct interaction between the learner and the environment to gather data from both parties \cite{xue2021inverse, lian2021robust, perrusquia2022complementary}. While these methods offer flexibility, they are bi-level approaches that often demand complex solutions, involving the nested resolution of both forward and inverse optimal control problems. Additionally, they might not be applicable to multi-dimensional input systems.

In \cite{ASL_NN2024}, we proposed a model-free approach that eliminates the need for solving the forward problem at each iteration, thereby reducing complexity and facilitating its use in online applications where expert trajectories dissipate rapidly. However, similar to \cite{kamalapurkar2018linear, self2019online, SELF2022110242}, in this method a finite informativity condition, that is necessary for the solution, may not be satisfied for certain systems. An example of such systems includes a class of linear systems with quadratic cost functions where the penalty weight matrices are not diagonal. Although the model-free approach presented in \cite{perrusquia2023drone} does not have this limitation and, unlike \cite{xue2021inverse, lian2021robust, perrusquia2022complementary}, can be applied to multi-dimensional input systems, it still necessitates solving the forward problem in a nested manner, similar to \cite{xue2021inverse, lian2021robust, perrusquia2022complementary}.

In this paper, we present a model-free and a partially model-free IRL/IOC algorithm for estimating the cost function of input-affine continuous-time nonlinear deterministic systems. The proposed approaches aim to address the finite informativity issue present in the previous works---thereby extending the applicability of our methods to a broader class of input-affine systems---and to reduce the complexity found in the bi-level model-free methods. These goals are achieved by estimating a subset of unknown parameters using policy information and then estimating the remaining parameters utilizing the Hamilton-Jacobi-Bellman (HJB) equation after the convergence of the first set. This approach differs from those in \cite{kamalapurkar2018linear, SELF2022110242, ASL_NN2024}, which estimate all the parameters simultaneously. The model-free approach utilizes both expert and learner agent trajectories, while the partially model-free approach, which is a special case of the model-free approach for systems with known input dynamics, relies solely on the expert's trajectories. The algorithms were evaluated through numerical simulations in MATLAB under both noise-free and noisy signal conditions, as well as in a physically realistic MuJoCo environment, with results indicating satisfactory performance across these settings. The model-free algorithm proposed in this paper is an extension of our previous design for linear systems \cite{ASL_ESWA2024}, adapted for nonlinear systems.

The main contributions of this study, in comparison to existing works in the literature, are as follows:\begin{enumerate}
\item In contrast to the designs in \cite{kamalapurkar2018linear, self2019online, SELF2022110242, ASL_NN2024}, which may not satisfy the required finite informativity condition for certain input-affine systems, our proposed methods are applicable to a broader class of such systems.
\item Unlike the bi-level methods in \cite{xue2021inverse, lian2021robust, perrusquia2022complementary, perrusquia2023drone}, our algorithms avoid the need to solve a forward optimal control or reinforcement learning problem after each cost function update. The forward problem is only necessary during the initialization phase of the model-free algorithm.
\item Compared to the approaches in \cite{xue2021inverse, lian2021robust, perrusquia2022complementary}, our methods are suitable for multi-dimensional input systems that require the estimation of input-penalty weights.
\item In contrast to \cite{xue2021inverse, lian2021robust, lian2021online, lian2022inverse, perrusquia2022complementary, lian2024inverse}, our algorithms do not apply the updated policies within the estimation process to the system at each iteration, which helps mitigate concerns about the stability of these updated policies.
\end{enumerate}

\section{Inverse optimal control problem definition}
We have considered an expert system with the following
input-affine continuous-time nonlinear dynamic model:
\begin{align}
\dot{\mathbf{x}}_e(t)=\mathbf{f}(\mathbf{x}_e)+\mathbf{g}(\mathbf{x}_e)\mathbf{u}_e(t),
\label{system}
\end{align}
where $\mathbf{x}_e\in\mathbb{R}^{n}$ is the state vector, $\mathbf{u}_e\in\mathbb{R}^{m}$ is the input vector, $\mathbf{f}(\cdot)\in\mathbb{R}^{n}$ represents drift dynamics, and $\mathbf{g}(\cdot)\in\mathbb{R}^{n\times m}$ is input dynamics, which is assumed that its transpose can be represented as follows:
\begin{equation}
\mathbf{g}(\mathbf{x}_e)^{\top}=\mathbf{W}_g\boldsymbol{\sigma}_g(\mathbf{x}_e)+\boldsymbol{\epsilon}_g,
\label{g_NN}
\end{equation}
where $\mathbf{W}_g\in\mathbb{R}^{m\times \mathcal{L}_g}$ is a constant matrix for some positive integer $\mathcal{L}_g$, $\boldsymbol{\sigma}_g\in\mathbb{R}^{\mathcal{L}_g \times n}$ is the known continuously differentiable basis matrix with locally Lipschitz continuous gradien, and $\boldsymbol{\epsilon}_g\in\mathbb{R}^{m\times n}$ is the functional reconstruction error. We have also considered a learner system with the same dynamic model as the expert system, given by $\dot{\mathbf{x}}_l(t)=\mathbf{f}(\mathbf{x}_l)+\mathbf{g}(\mathbf{x}_l)\mathbf{u}_l(t)$, where $\mathbf{x}_l\in\mathbb{R}^{n}$ and $\mathbf{u}_l\in\mathbb{R}^{m}$ are its state and input vectors, respectively.

It is assumed that the expert agent executes a policy that minimizes the following cost:
\begin{equation}
J(\mathbf{x}_e(0),\mathbf{u}_e)\triangleq\int_0^{\infty}[Q(\mathbf{x}_e(t))+\mathbf{u}_e(t)^{\top}\mathbf{R}\mathbf{u}_e(t)]\text{d}t,
\label{J}
\end{equation}
where $Q(\mathbf{x}_e)\in\mathbb{R}$ is a continuously differentiable and positive-semidefinite function that penalizes state trajectories and $\mathbf{R}\in\mathbb{R}^{m\times m}$ is a positive-definite matrix, which determines input penalties. It is assumed that $Q(\mathbf{x}_e)$ can be represented as
\begin{equation}
Q(\mathbf{x}_e)=\mathbf{W}_Q^{\top}\boldsymbol{\sigma}_Q(\mathbf{x}_e)+\epsilon_Q,
\label{Q_NN}
\end{equation}
where $\boldsymbol{\sigma}_Q\in\mathbb{R}^{\mathcal{L}_Q}$ is the known basis vector for some positive integer $\mathcal{L}_Q$, $\mathbf{W}_Q\in\mathbb{R}^{\mathcal{L}_Q}$ is the weight vector, and $\epsilon_Q\in\mathbb{R}$ is the functional reconstruction error.

The objective of the IOC/IRL problem is to estimate $Q(\cdot)$ and $\mathbf{R}$ of the cost function (\ref{J}). It is widely recognized in the literature that this problem does not have a unique solution. However, any solution is considered acceptable as long as they enable the retrieval of the expert's optimal policy. Consequently, the estimated terms may differ from the true ones, resulting in what are commonly referred to as equivalent cost functions.

Under agent's optimal decisions, i.e., when $\mathbf{u}_e$ minimizes (\ref{J}), the states and optimal inputs satisfy the following Hamilton–Jacobi–Bellman (HJB) equation \cite{kamalapurkar2018linear}:
\begin{align}
H(\mathbf{x}_e,\mathbf{u}_e,\nabla V)
&\triangleq Q(\mathbf{x}_e)+\mathbf{u}_e^{\top}\mathbf{R}\mathbf{u}_e+\nabla V^{\top}\dot{\mathbf{x}}_e(t)=0,
\label{HJB_equation}
\end{align}
where $\nabla$ denotes the gradient with respect to $\mathbf{x}_e$ and $V\in\mathbb{R}$ is the optimal value function defined as follows:
\begin{equation}
V(\mathbf{x}_e(t))\triangleq\min_{\mathbf{u}_e}\int_t^{\infty}[Q(\mathbf{x}_e(\tau))+\mathbf{u}_e(\tau)^{\top}\mathbf{R}\mathbf{u}_e(\tau)]\text{d}\tau.
\end{equation}
The value function can be represented in parametric form as
\begin{equation}
V=\mathbf{W}_V^{\top}\boldsymbol{\sigma}_V(\mathbf{x}_e)+\epsilon_V,
\label{V_NN}
\end{equation}
where $\mathbf{W}_V\in\mathbb{R}^{\mathcal{L}_V}$ is the weight vector for some positive integer $\mathcal{L}_V$, $\boldsymbol{\sigma}_V\in\mathbb{R}^{\mathcal{L}_V}$ is the known continuously differentiable basis vector with locally Lipschitz continuous gradien, and $\epsilon_V\in\mathbb{R}$ is the error function \cite{modares2013adaptive}.

The closed-form equation of the optimal controller that minimizes (\ref{J}) can be written as follows \cite{asl2023online}:
\begin{equation}
%\mathbf{u}_e
%=-\frac{1}{2}\mathbf{R}^{-1}\mathbf{g}(\mathbf{x}_e)^{\top}(\frac{\partial V}{\partial \mathbf{x}_e})^{\top}
\mathbf{u}_e
=-\frac{1}{2}\mathbf{R}^{-1}\mathbf{g}(\mathbf{x}_e)^{\top}\nabla V.
\label{u_e}
\end{equation}

The HJB equation (\ref{HJB_equation}) and the closed-form policy equation (\ref{u_e}) are two key relations commonly used in the literature to solve the IOC/IRL problem. The latter provides specific information unique to input-affine deterministic systems, enabling solutions with features that differ from those developed for stochastic systems. Different methods of using these relations have led to various solutions.

The approaches in \cite{kamalapurkar2018linear, self2019online, SELF2022110242, self2020online, Asl_IRL22} derive two performance metrics using (\ref{HJB_equation}) and (\ref{u_e}). While these methods benefit from not requiring the interaction of the learner system with the environment and rely solely on the expert's data, they require knowledge of the input dynamics $\mathbf{g}(\cdot)$, as seen in (\ref{u_e}). Furthermore, these methods are applicable only to a limited class of systems. Specifically, as we demonstrated for the multi-player case in \cite{Asl_SCL2024}, these methods require satisfaction of a finite informativity condition, which can only be met in certain systems.

The solutions developed in \cite{xue2021inverse_tracking, lian2021robust, lian2021inverse, lian2021online, lian2022inverse} do not require knowledge of the input dynamics and are model-free. This feature is primarily achieved by utilizing the trajectories of the learner system, meaning that the learner must interact with the environment in these methods. However, these methods are bi-level and, therefore, more complex, as they require solving both forward and inverse optimal control problems in a nested manner. Hence, these methods lead to increased computational complexity and potential stability concerns when repeatedly solving the forward optimal control problem. Additionally, as noted in \cite{ASL_ESWA2024}, because these methods assign an arbitrary value to the input weight matrix $\mathbf{R}$, they may not be applicable to multi-dimensional input systems.

In the subsequent section, we will utilize (\ref{HJB_equation}) and (\ref{u_e}) to propose an algorithm that alleviates the aforementioned issues in the existing literature. In particular, our proposed method eliminates the need for a model of the system, extends applicability to a broader class of systems, and significantly reduces computational complexity by removing the requirement to solve the forward and inverse problems in a nested manner.

\section{Model-free solution}
In this section, we present a model-free algorithm to estimate either the true or equivalent versions of the cost function (\ref{J}). Unlike the methods in \cite{kamalapurkar2018linear, self2019online, Asl_IRL22}, which estimate all unknown parameters of the cost function simultaneously---potentially failing to satisfy the required finite informativity condition for some systems---the proposed method in this section divides the unknown parameters into two sets and estimates them separately.

The first set of parameters, which includes the input weight matrix $\mathbf{R}$ and the intermediary variable $\mathbf{W}_{V}$ that defines the value function through (\ref{V_NN}), is estimated iteratively through a gradient descent approach using the error between the optimal control inputs of the expert agent and the control inputs derived from the current estimate of the parameters. The second set of parameters, which includes the parameters of the state penalty function $Q_l(\cdot)$, is estimated after the convergence of the first set by utilizing the HJB equation. This technique guarantees the satisfaction of the required finite informativity condition for all systems defined by (\ref{system}). Additionally, by using the closed-form equation of the optimal policy in the estimation of $\mathbf{R}$ and $\mathbf{W}_{V}$ and updating these quantities together without requiring updates to $Q_l(\cdot)$, we eliminate the need for the complex data-driven policy updates at each iteration that are required in \cite{xue2021inverse, lian2021robust, perrusquia2022complementary, perrusquia2023drone}. Another important advantage of the proposed algorithm, compared to these and similar works, is that the policy updated during the estimation of the first set of parameters is not applied to the learner system, alleviating concerns about ensuring the stability of the updated policy.

\subsection{Estimation of the first set of parameters: using the gradient descent method}
In this section, we present an iterative method to estimate the input-penalty weight matrix $\mathbf{R}$ and $\mathbf{W}_{V}$ that defines the value function $V(\cdot)$ using the gradient descent technique. Although the value function is not directly required in the definition of the cost function (\ref{J}), its estimation serves as an intermediary variable essential to our algorithm's operation.

The key information for the estimation of $\mathbf{R}$ and $\mathbf{W}_{V}$ is the error between the control inputs of the expert agent $\mathbf{u}_e$ and the inputs of the learner agent $\mathbf{u}_l$, both obtained at the same states. In the proposed method, $\mathbf{u}_l$ can be obtained from the current estimate of the cost parameters without the need for %measurements from the learner agent or
 forcing the learner to factually visit the same states of the expert.

In order to develop model-free update laws, we first need to write the closed-form equation of the policies in the form of a known basis functions vector multiplied by an unknown constant matrix gain. This representation, which is also used in the model-free reinforcement learning (RL) methods for nonlinear systems, e.g., \cite{lee2014integral}, will be utilized to write the parameter vector of input dynamics $\mathbf{W}_g$, introduced in (\ref{g_NN}), based on the previous estimate of the policy gain and other estimated parameters, facilitating a model-free estimation.

To derive the closed-form equation of the optimal policy in the desired form, first, using (\ref{g_NN}) and (\ref{V_NN}) and assuming that $\boldsymbol{\epsilon}_g$ and $\epsilon_V$ are negligible, we develop the term $\mathbf{g}^{\top}\nabla V$, that appears in the optimal controller (\ref{u_e}), in the following form:
\begin{align}
\mathbf{g}(\mathbf{x}_e)^{\top}\nabla V(\mathbf{x}_e)&=\mathbf{W}_g\boldsymbol{\sigma}_g(\mathbf{x}_e)\nabla\boldsymbol{\sigma}_V(\mathbf{x}_e)^{\top}\mathbf{W}_V\nonumber\\
&=\mathbf{W}_g\mathbf{W}_u\boldsymbol{\sigma}_u(\mathbf{x}_e),
\label{g_GradV}
\end{align}
in which $\boldsymbol{\sigma}_u(\cdot)\in\mathbb{R}^{\mathcal{L}_u}$ is the basis vector for some positive integer $\mathcal{L}_u$, and $\mathbf{W}_u\in\mathbb{R}^{\mathcal{L}_g\times\mathcal{L}_u}$ is the matrix that is developed from the elements of $\mathbf{W}_V$. Note that the basis vector function $\boldsymbol{\sigma}_u$ contains unique and non-zero terms. Using (\ref{u_e}) and (\ref{g_GradV}), we can develop $\mathbf{u}_e$ as the product of a constant gain matrix, containing all unknown parameters, and a known basis vector as $\mathbf{u}_e=-\mathbf{K}_e\boldsymbol{\sigma}_u(\mathbf{x}_e)$, where the optimal control gain $\mathbf{K}_e\in\mathbb{R}^{m\times\mathcal{L}_u}$ is defined as $\mathbf{K}_e\triangleq 0.5\mathbf{R}^{-1}\mathbf{W}_g\mathbf{W}_u$. This representation of the optimal policy will be mainly utilized for the development of a model-free update law for $\mathbf{W}_V$.

We consider the same representation for the policy of the learner agent. That is, we assume that $\mathbf{K}_{l}\in\mathbb{R}^{m\times n}$ is an optimal control gain obtained for a specific cost function defined by the positive-semidefinite function $Q_l(\cdot)\in\mathbb{R}^{n\times n}$ and the matrix $\mathbf{R}_l\in\mathbb{R}^{m\times m}$, where $Q_l(\cdot)$ and $\mathbf{R}_l$ are, respectively, corresponding terms to $Q(\cdot)$ and $\mathbf{R}$ in (\ref{J}). Hence, the optimal controller and the control gain matrix can be defined as
\begin{align}
&\mathbf{u}_l=-\mathbf{K}_l\boldsymbol{\sigma}_u(\mathbf{x}_l),\label{u_l}\\
&\mathbf{K}_l\triangleq \frac{1}{2}\mathbf{R}_l^{-1}\mathbf{W}_g\mathbf{W}_{ul},
\label{K_l}
\end{align}
where $\mathbf{W}_{ul}\in\mathbb{R}^{\mathcal{L}_g\times \mathcal{L}_u}$ is a matrix that satisfies $\boldsymbol{\sigma}_g\nabla\boldsymbol{\sigma}_V^{\top}\mathbf{W}_{Vl}=\mathbf{W}_{ul}\boldsymbol{\sigma}_u$, in which $\mathbf{W}_{Vl}\in\mathbb{R}^{\mathcal{L}_V}$ is the weight vector defining the value function corresponding to this controller.

Considering the definition of $\mathbf{K}_l$, the difference between $\mathbf{K}_l$ and $\mathbf{K}_e$ can be used to estimate $\mathbf{R}_l$ and $\mathbf{W}_{ul}$ through a gradient descent method, similar to the approach in \cite{ASL_ESWA2024}. However, this method requires the estimation of $\mathbf{K}_e$. In addition, since the true $\mathbf{W}_{ul}$ may contain zero and/or identical components, the estimation using these quantities through the gradient descent method may require techniques such as masking that might increase the complexity of the estimation process. Therefore, for simplicity, we utilize the difference between control inputs $\mathbf{u}_l$ and $\mathbf{u}_e$ to estimate $\mathbf{R}_{l}$ and $\mathbf{W}_{Vl}$. Although the control inputs may have less information compared to the control gains (due to reduced dimension), the estimation is still feasible by applying corrections over multiple data points, as detailed in the following.

The difference between $\mathbf{u}_l$ and $\mathbf{u}_e$ is defined as the error vector $\mathbf{E}_{u}\in\mathbb{R}^{m}$, which is given by
\begin{equation}
\mathbf{E}_{u}\triangleq\mathbf{u}_l-\mathbf{u}_e.
\label{E_K}
\end{equation}
This error signal can be obtained by using the observation data $\mathbf{u}_e$ and calculating $\mathbf{u}_l$ at the corresponding state, $\mathbf{x}_e$, using (\ref{u_l}). In (\ref{u_l}), $\mathbf{K}_l$ can be estimated using model-free approaches, such as the RL techniques detailed in \cite{lee2014integral} for some ${Q}_l(\cdot)$ and $\mathbf{R}_l$. Although the estimation of $\mathbf{K}_l$ through model-free methods is computationally expensive, as it will be cleared later, compared to the methods in  \cite{xue2021inverse_tracking, lian2021robust, lian2021inverse, lian2021online, lian2022inverse, perrusquia2023drone}, this process is only required once in the initialization of the algorithm.

Using the error vector (\ref{E_K}) and calculating this quantity at $i\in \{1,\cdots\mathcal{N}\}$ different state points, we define a scalar distance function that quantifies the disparity between the policies as follows:
\begin{equation}
\mathcal{E}\triangleq\sum_{i=1}^{\mathcal{N}}(\mathbf{E}_{u}^{(i)})^{\top}\mathbf{E}_{u}^{(i)},
\label{mathcal_E}
\end{equation}
where the superscript $(i)$ denotes the $i$-th sample of the error.

Using the optimal policy relation in (\ref{u_e}) and the parametric form of the value function and input dynamics (neglecting the errors $\boldsymbol{\epsilon}_g$ and $\epsilon_V$), the controller $\mathbf{u}_l$ can also be written in the following form: 
\begin{equation}
\mathbf{u}_l=-\frac{1}{2}\mathbf{R}_l^{-1}\mathbf{W}_g\boldsymbol{\sigma}_g\nabla\boldsymbol{\sigma}_V^{\top}\mathbf{W}_{Vl}.
\label{u_l_W_V}
\end{equation}
In this relation, two variables, $\mathbf{R}_l$ and $\mathbf{W}_{Vl}$, can be tuned to minimize $\mathcal{E}$. However, since $\mathcal{E}$ is not simultaneously convex with respect to both variables, the minimization process might converge to a local minimum, which is acceptable only if $\mathcal{E}$ is reduced below a specified threshold. Additionally, even if a combination of $\mathbf{R}_l$ and $\mathbf{W}_{Vl}$ that minimizes (\ref{mathcal_E}) is identified, there is no assurance that this set will yield a positive-semidefinite function $Q_l(\cdot)$, which will be estimated in the subsequent section. To overcome these challenges, the algorithm restarts the estimation process by assigning new initial values to $Q_l(\cdot)$ and $\mathbf{R}_l$ if the necessary conditions are not satisfied.

For the estimation of $\mathbf{R}_l$ and $\mathbf{W}_{Vl}$, the algorithm requires initial random estimates for $Q_l(\cdot)\geq 0$ and $\mathbf{R}_l\succ0$. The policy gain $\mathbf{K}_l$ and the value function weight vector $\mathbf{W}_{Vl}$ can then be obtained, for example, through model-based RL methods such as those detailed in \cite{lee2014integral}. The next steps involve developing $\mathbf{W}_{ul}$ from $\mathbf{W}_{Vl}$ such that $\boldsymbol{\sigma}_g\nabla\boldsymbol{\sigma}_V^{\top}\mathbf{W}_{Vl}=\mathbf{W}_{ul}\boldsymbol{\sigma}_u$ is satisfied, and estimating $\mathbf{R}_l$ and $\mathbf{W}_{Vl}$ such that the distance function (\ref{mathcal_E}) is minimized.

Considering (\ref{E_K}) and (\ref{u_l_W_V}), different methods can be applied for updating $\mathbf{R}_l$ and $\mathbf{W}_{Vl}$ to minimize $\mathcal{E}$. Here, we utilize a stochastic gradient descent technique, where the matrices are iteratively updated using the gradient of the distance function obtained at a single, randomly chosen state point, as follows:
%Compared to [Asl and Drone], in order to avoid estimating the expert's policy, we apply stochastic gradient decent technique for updating R and W. That is, we update them by using current sample XX.
%Subsequently, using (\ref{K_l}) and (\ref{E_K}), the matrices $\mathbf{R}_l$ and $\mathbf{W}_{Vl}$ are updated using the following gradient descent update laws:
\begin{align}
%&\mathbf{R}_{l}^{(k+1)}=\mathbf{R}_{l}^{(k)}\nonumber\\
%&+\alpha_{R}\left((\mathbf{R}_{l}^{(k)})^{-1}\mathbf{K}_{l}^{(k)}\boldsymbol{\sigma}_u(\mathbf{E}_{u}^{(k)})^{\top}+\mathbf{E}_{u}^{(k)}(\mathbf{K}_{l}^{(k)}\boldsymbol{\sigma}_u)^{\top}(\mathbf{R}_{l}^{(k)})^{-1}\right), \label{R_l_k}\\
&\mathbf{R}_{l}^{(k+1)}=\mathbf{R}_{l}^{(k)}\nonumber\\
&+\alpha_{R}\left((\mathbf{R}_{l}^{(k)})^{-1}\mathbf{u}_{l}(\mathbf{E}_{u}^{(k)})^{\top}+\mathbf{E}_{u}^{(k)}\mathbf{u}_{l}^{\top}(\mathbf{R}_{l}^{(k)})^{-1}\right), \label{R_l_k}\\
&\mathbf{W}_{Vl}^{(k+1)}=\mathbf{W}_{Vl}^{(k)}+\alpha_{V}\left(\nabla\boldsymbol{\sigma}_V\boldsymbol{\sigma}_g^{\top}((\mathbf{W}_{ul}^{(k)})^{\dagger})^{\top}(\mathbf{K}_{l}^{(k)})^{\top}\mathbf{E}_{u}^{(k)}\right), \label{W_Vl_k}
\end{align}
where $\dagger$ denotes the pseudoinverse, $\alpha_R$ and $\alpha_V$ represent small positive gains, the superscript $(k)$ denotes the $k$-th iteration, and the arguments of the functions $\boldsymbol{\sigma}_g(\cdot)$ and $\nabla\boldsymbol{\sigma}_V(\cdot)$ are the state vector $\mathbf{x}_e$ corresponding to the sampled $\mathbf{u}_e$. Note that $\mathbf{W}_{ul}^{(k+1)}$ can be obtained from $\mathbf{W}_{Vl}^{(k+1)}$ to be used in the next iteration in (\ref{W_Vl_k}). Also, to derive (\ref{W_Vl_k}), we have utilized (\ref{u_l_W_V}) and the following relation:
\begin{equation}
2\mathbf{K}_l\mathbf{W}_{ul}^{\dagger}=\mathbf{R}_{l}^{-1}\mathbf{W}_g.
\end{equation}
The update gains $\alpha_R$ and $\alpha_u$ are chosen to be sufficiently small to ensure a step size toward the direction of the negative gradient. Given the necessity for the matrix $\mathbf{R}_l$ and the function $\mathbf{W}_{Vl}^{\top}\boldsymbol{\sigma}_V(\cdot)$ to remain positive-definite, the step sizes are selected to uphold this condition. It is worth noting that alternative advanced optimization techniques, such as projection-based methods, and Cholesky factorization-based updates, can also be employed to rigorously enforce positive definiteness during parameter updates \cite{boyd2004convex, nocedal1999numerical}.

For the subsequent iteration of (\ref{R_l_k}) and (\ref{W_Vl_k}), an updated $\mathbf{K}_l$ is required, which can be obtained using the following model-free expression:
\begin{align}
&\mathbf{K}_{l}^{(k+1)}
=0.5(\mathbf{R}_{l}^{(k+1)})^{-1}\mathbf{W}_g\mathbf{W}_{ul}^{(k+1)}\nonumber\\
&=0.5(\mathbf{R}_{l}^{(k+1)})^{-1}\mathbf{R}_{l}^{(k)}(\mathbf{R}_{l}^{(k)})^{-1}\mathbf{W}_g\mathbf{W}_{ul}^{(k)}(\mathbf{W}_{ul}^{(k)})^{\dagger}\mathbf{W}_{ul}^{(k+1)}\nonumber\\
&=(\mathbf{R}_{l}^{(k+1)})^{-1}\mathbf{R}_{l}^{(k)}\mathbf{K}_{l}^{(k)}(\mathbf{W}_{ul}^{(k)})^{\dagger}\mathbf{W}_{ul}^{(k+1)}.
%&=0.5(\mathbf{R}_{l}^{(k)}+\boldsymbol{\Delta}_R)^{-1}\mathbf{W}_g\mathbf{W}_{ul}^{(k+1)}\nonumber\\
%&=0.5(\mathbf{I}_{m}+(\mathbf{R}_{l}^{(k)})^{-1}\boldsymbol{\Delta}_R)^{-1}(\mathbf{R}_{l}^{(k)})^{-1}\mathbf{W}_g\mathbf{W}_{ul}^{(k+1)}\nonumber\\
%&=(\mathbf{I}_{m}+(\mathbf{R}_{l}^{(k)})^{-1}\boldsymbol{\Delta}_R)^{-1}\mathbf{K}_{l}^{(k)}(\mathbf{W}_{ul}^{(k)})^{\dagger}\mathbf{W}_{ul}^{(k+1)}.
\label{K_next}
\end{align}
%in which $\mathbf{W}_{ul}^{(k+1)}$ can be obtained from $\mathbf{W}_{Vl}^{(k+1)}$. 
The updated $\mathbf{K}_l$ will be directly used in (\ref{W_Vl_k}). It is also required to calculate $\mathbf{u}_l$ in the next selected state vector through (\ref{u_l}) for use in (\ref{R_l_k}) and the calculation of $\mathbf{E}_u$.

The variables $\mathbf{R}_l$ and $\mathbf{W}_{Vl}$ are iteratively updated through (\ref{R_l_k}) and (\ref{W_Vl_k}) until $\mathcal{E}^{(k)}$ becomes less than $\epsilon_E$, for some small $\epsilon_E\in\mathbb{R}^{+}$, for a sufficient number of last iterations. Note that the updates presented in (\ref{R_l_k}) and (\ref{W_Vl_k}) constitute a gradient-descent-based optimization to minimize the quadratic error cost function (\ref{mathcal_E}). Since $\mathcal{E}$ is constructed as a quadratic sum of smooth errors, it is continuously differentiable with Lipschitz-continuous gradients. Then, with sufficiently small step sizes, $\alpha_R$ and $\alpha_V$, the gradient descent update ensures a monotone decrease in $\mathcal{E}$ at each iteration, guaranteeing the convergence of the estimated parameters to a locally optimal solution.

Since we update both $\mathbf{R}_l$ and $\mathbf{W}_{Vl}$ in this step, the policy gain $\mathbf{K}_l$ corresponding to the updated cost function can be obtained through (\ref{K_next}) regardless of the value of $Q_l(\cdot)$. This approach, compared to the methods in \cite{xue2021inverse_tracking, lian2021robust, lian2021inverse, lian2021online, lian2022inverse, perrusquia2023drone}, offers two distinct advantages. First, in the subsequent iterations of $\mathbf{R}_l$ and $\mathbf{W}_{Vl}$, there is no need to solve a forward RL problem. Second, the process is offline, and hence, the updated policy $\mathbf{K}_{l}$ in (\ref{K_next}) does not need to be applied to the system, eliminating concerns regarding the stability of the updated policy. 

\subsection{Estimation of the second set of parameters: using the HJB equation}
In this section, we present a method to estimate the state-penalty function $Q_l(\cdot)$ using the HJB equation, along with the known terms $\mathbf{R}_l$ and $\mathbf{W}_{Vl}$ estimated in the previous section. Since the control gain $\mathbf{K}_l$ has converged to the optimal one $\mathbf{K}_e$ in the previous step, we can utilize the trajectories of the expert agent to develop the HJB equation for estimating $Q_l(\cdot)$. This is possible because the same controllers yield the same performance. Using the expert's trajectories simplifies the estimation procedure, as the observation data from the expert, utilized in the previous step, can be directly applied here to develop the HJB equation. This approach eliminates the need to further run the learner/expert agents for the estimation of $Q_l(\cdot)$. Following this approach, the learner agent is only required to be run initially to obtain values for $\mathbf{K}_{l}$ and $\mathbf{W}_{Vl}$ for the first iteration of (\ref{R_l_k}), (\ref{W_Vl_k}), and (\ref{K_next}).

In order to facilitate a model-free estimation of $Q_l(\cdot)$, similar to integral RL techniques \cite{modares2014integral, lee2014integral}, we integrate the HJB equation (\ref{HJB_equation}) over any time interval $T>0$, extending from $t-T$ to $t$, and use the estimated terms $\mathbf{R}_l$ and $\mathbf{W}_{Vl}$ from the previous section to develop the following relation:
\begin{align}
V_l(\mathbf{x}_e(t))=
&~V_l(\mathbf{x}_e(t-T))\nonumber\\
&-\int_{t-T}^{t}[Q_l(\mathbf{x}_e(\tau))+\mathbf{u}_{e}(\tau)^{\top}\mathbf{R}_{l}\mathbf{u}_{e}(\tau)]\text{d}\tau,
\label{HJB_int}
\end{align}
where, using (\ref{V_NN}) and (\ref{Q_NN}), $V_l(\cdot)$ and $Q_l(\cdot)$ are defined as $V_l(\cdot)\triangleq\mathbf{W}_{Vl}^{\top}\boldsymbol{\sigma}_V(\cdot)+\epsilon_V$, and $Q_l(\cdot)\triangleq\mathbf{W}_{Ql}^{\top}\boldsymbol{\sigma}_Q(\cdot)+\epsilon_Q$, for some $\mathbf{W}_{Ql}\in\mathbb{R}^{\mathcal{L}_Q}$. The integral HJB equation (\ref{HJB_int}) is the key relation for the estimation of $Q_l(\cdot)$. Note that this equation is developed using the trajectories of the expert system; that is, $(\mathbf{x}_e,\mathbf{u}_e)$. 

Now, neglecting the functional reconstruction errors, we can write (\ref{HJB_int}) in the parametric form as follows:
\begin{equation}
\mathbf{W}_{Ql}^{\top}\bar{\boldsymbol{\sigma}}_Q
=\psi,
\label{psi}
\end{equation}
where the known functions $\bar{\boldsymbol{\sigma}}_Q\in\mathbb{R}^{\mathcal{L}_Q}$ and $\psi\in\mathbb{R}$ are defined as
\begin{align}
&\bar{\boldsymbol{\sigma}}_Q(\mathbf{x}_e)
\triangleq\int_{t-T}^{t}\boldsymbol{\sigma}_Q(\mathbf{x}_e(\tau))\text{d}\tau, \label{bar_sigma_Q}\\
&\psi\triangleq V_l(\mathbf{x}_e(t-T))-V_l(\mathbf{x}_e(t))-\int_{t-T}^{t}\mathbf{u}_{e}(\tau)^{\top}\mathbf{R}_{l}\mathbf{u}_{e}(\tau)\text{d}\tau.
\end{align}
Then, the function $Q_l(\cdot)$ can be estimated using a least-squares method to solve the linear system (\ref{psi}) for the unknown vector $\mathbf{W}_{Ql}$ by collecting a sufficient amount of data.
\begin{definition}
The signal $\mathbf{x}_e$ is called finitely informative if
there exist time instances $0\leq t_1<\cdots<t_N$ such that the resulting history stack of $\bar{\boldsymbol{\sigma}}_Q(\mathbf{x}_e)$ is full rank.
\end{definition}

For the system (\ref{psi}), it is guaranteed that a sufficient collection of data results in a full-rank history stack (i.e., satisfies the finite informativity condition), allowing the linear system (\ref{psi}) to admit a unique solution. Therefore, this approach has an advantage over the techniques in \cite{kamalapurkar2018linear, self2019online, Asl_IRL22}, where the history stack might not be full rank for certain systems \cite{ASL_ESWA2024}.

The complete steps of the algorithm are summarized in Algorithm~\ref{Alg}, and the block diagram of the proposed algorithm is shown in Fig.~\ref{BlockDiagram_Alg1}. The algorithm can be slightly simplified for single-dimensional input systems, as indicated by the following lemma.
\begin{figure}
	\centering
	\includegraphics[scale=.5]{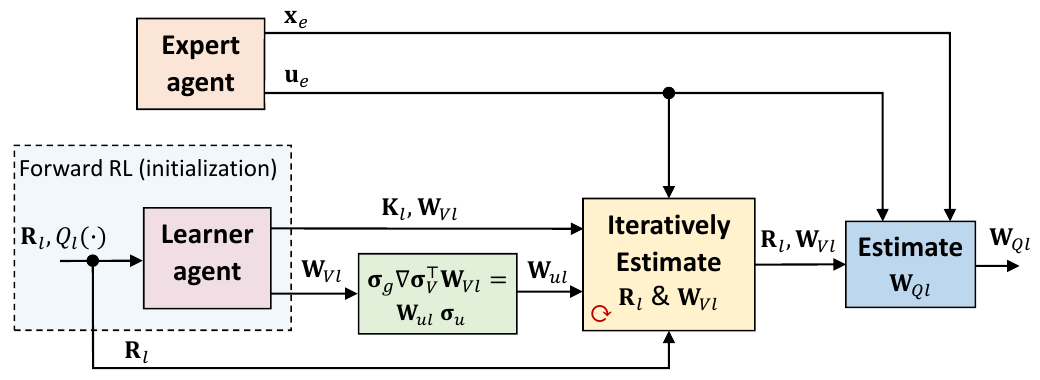}
	\caption{Block diagram of the proposed IRL/IOC method (Algorithm~\ref{Alg}).}
	\label{BlockDiagram_Alg1}
\end{figure}
\begin{algorithm}
\caption{Model-free IRL/IOC approach}\label{Alg}
\begin{algorithmic}[1]
\item Collect input-state trajectories of the expert;
\item Set the iteration number $k=1$ and assign initial values $Q_l(\cdot)\geq 0$ and $\mathbf{R}_{l}^{(k)}\succ0$; then solve for $\mathbf{W}_{Vl}^{(k)}$ and $\mathbf{K}_{l}^{(k)}$ through a model-free RL method on the learner side;\label{Step_initialize_2}
\item Develop $\mathbf{W}_{ul}$ from $\mathbf{W}_{Vl}$ such that $\boldsymbol{\sigma}_g\nabla\boldsymbol{\sigma}_V^{\top}\mathbf{W}_{Vl}=\mathbf{W}_{ul}\boldsymbol{\sigma}_u$ is satisfied; \label{Devel_W_ul}
\item Estimate $\mathbf{R}_l$ and $\mathbf{W}_{Vl}$;
\label{Step_Update_R_P}
\newline
Loop over $k$:
\begin{itemize}
\setlength\itemsep{.0em}
\item Select a sample from the expert's data $(\mathbf{x}_e,\mathbf{u}_e)$;
\item Calculate $\mathbf{u}_l$ at the selected states $\mathbf{x}_e$ using (\ref{u_l});
\item Calculate the policy difference: $\mathbf{E}_{u}^{(k)}=\mathbf{u}_{l}-\mathbf{u}_e$;
\item Stop if the recent values of $\mathcal{E}^{(k)}$ over the last several iterations are all less than $\epsilon_E$;
\item Update $\mathbf{R}_{l}$ through (\ref{R_l_k});
\item Update $\mathbf{W}_{Vl}$ through (\ref{W_Vl_k});
\item Update $\mathbf{K}_l$ through (\ref{K_next});
\end{itemize}
\item Estimate $Q_l(\cdot)$ using (\ref{psi}); \label{Step_Update_Q}
\item If $Q_l(\cdot)\leq 0$, go to Step~\ref{Step_initialize_2}.
\end{algorithmic}
\end{algorithm}

\begin{lemma}\label{Lemma_R}
For systems with single-dimensional input, an arbitrary fixed positive value can be assigned to $\mathbf{R}_l$. Consequently, the update of $\mathbf{R}_l$ in Step~\ref{Step_Update_R_P} of Algorithm~\ref{Alg} is not required when $m=1$.
\end{lemma}
\begin{proof}
The steps of the proof are similar to those presented in \cite{ASL_ESWA2024} for linear systems.
\end{proof}
\begin{remark}
For systems with $m>1$, we may not be able to select $\mathbf{R}_l$ arbitrarily and will need to update it in Step~\ref{Step_Update_R_P} of Algorithm~\ref{Alg}. An example system is provided in \cite{ASL_ESWA2024}.
\end{remark}

\subsection{Estimation of the second set of parameters: using the enhanced HJB equation}
Although the finite informativity condition can be guaranteed for the system (\ref{psi}) with a sufficient collection of data, the estimation error of $Q_l(\cdot)$ (or $\mathbf{W}_{Ql}$) could be large when the state trajectories evolve with low amplitudes, due to the overlooked functional reconstruction error terms $\epsilon_Q$ and $\epsilon_V$ in the development of (\ref{psi}). This is because, under these conditions, the relative influence of the functional reconstruction errors becomes significant in (\ref{psi}). Furthermore, similar to the problem of estimating the value function in forward RL \cite{modares2014integral}, as the system's dimension and the number of basis functions for $Q_l(\cdot)$ increase, richer data might be required for an accurate estimation of $Q_l(\cdot)$. If the collected data are not sufficiently rich, the minimum singular value of the learning matrix obtained from $\bar{\boldsymbol{\sigma}}_Q(\cdot)$ will be small, which may lead to reduced accuracy in the estimation of $Q_l(\cdot)$.

To alleviate the aforementioned problems, we apply off-policy inputs to the learner agent in an effort to enrich the signals and increase their amplitude, and develop the enhanced HJB equation using the trajectories of the learner system as follows:
\begin{align}
&V_l(\mathbf{x}_l(t))=
V_l(\mathbf{x}_l(t-T))\nonumber\\
&-\int_{t-T}^{t}[Q_l(\mathbf{x}_l(\tau))+2\mathbf{u}_l^{\ast}(\tau)^{\top}\mathbf{R}_l\mathbf{u}_p(\tau)+\mathbf{u}_l^{\ast}(\tau)^{\top}\mathbf{R}_{l}\mathbf{u}_l^{\ast}(\tau)]\text{d}\tau,
\label{HJB_int_learner}
\end{align}
where $\mathbf{u}_p(t)\in\mathbb{R}^{m}$ is an enriching signal and $\mathbf{u}_l^{\ast}=-\mathbf{K}_e\boldsymbol{\sigma}_u(\mathbf{x}_l)$, and we have $\mathbf{u}_l=\mathbf{u}_l^{\ast}+\mathbf{u}_p$. The rationale for the inclusion of $2{\mathbf{u}_l^{\ast}}^{\top}\mathbf{R}_l\mathbf{u}_p$ in (\ref{HJB_int_learner}) can be found in \cite{ASL_NN2024}. We can use (\ref{HJB_int_learner}) to estimate $\mathbf{W}_{Ql}$ by developing its parametric form, similar to (\ref{psi}), where the known function $\bar{\boldsymbol{\sigma}}_Q$ is defined as in (\ref{bar_sigma_Q}) using the state trajectories of the learner $\mathbf{x}_l(t)$, and $\psi$ is defined as follows:
\begin{align}
\psi
&\triangleq V_l(\mathbf{x}_l(t-T))-V_l(\mathbf{x}_l(t))\nonumber\\
&-\int_{t-T}^{t}2\mathbf{u}_l^{\ast}(\tau)^{\top}\mathbf{R}_l\mathbf{u}_p(\tau)+\mathbf{u}_{l}^{\ast}(\tau)^{\top}\mathbf{R}_{l}\mathbf{u}_{l}^{\ast}(\tau)\text{d}\tau.
\end{align}

%==========================================
%==========================================
\section{Solution for systems with known input dynamics}
In this section, we address the solution of the IOC problem for the case where the input dynamics of the system $\mathbf{g}(\cdot)$ are known. We demonstrate that under this condition, Algorithm~\ref{Alg} can be modified to solve the problem using only the trajectories of the expert system, meaning that interaction between the learner agent and the environment is not required. This, along with other features that will be outlined subsequently, significantly simplifies the estimation procedure.  In comparison to the methods presented in \cite{kamalapurkar2018linear, self2019online, self2020online, SELF2022110242}, which also solve the IOC/IRL problem using only expert trajectories, our proposed approach can be applied to a wider range of systems, while all these methods require knowledge of the input dynamics. Furthermore, compared to Algorithm~3 in \cite{lian2022inverse}, which also assumes known input dynamics, our approach is much simpler, as the method in \cite{lian2022inverse} is bi-level and requires trajectories from both the expert and learner agents.

In Algorithm~\ref{Alg}, the trajectories of the learner system are primarily required in Step~\ref{Step_initialize_2}, where a model-free forward RL problem must be solved for an initially given cost function parameters $\mathbf{R}_l$ and $Q_l(\cdot)$. Before explaining why this step is unnecessary for systems with known input dynamics, we first present the modified update laws for the parameters specific to this condition.

When the input dynamics are known, the development of the controller in the form of (\ref{u_l}), which is typically used for the model-free update of $\mathbf{W}_{Vl}$, is no longer necessary. This feature allows us to eliminate Step~\ref{Devel_W_ul} of Algorithm~\ref{Alg}. In fact, by considering (\ref{V_NN}) and (\ref{u_e}), the iterative estimation of $\mathbf{R}_l$ and $\mathbf{W}_{Vl}$ can be performed using (\ref{R_l_k}) and the following relation:
\begin{align}
\mathbf{W}_{Vl}^{(k+1)}=\mathbf{W}_{Vl}^{(k)}+\alpha_{V}\left(\nabla\boldsymbol{\sigma}_V\mathbf{g}(\mathbf{R}_{l}^{(k)})^{-1}\mathbf{E}_{u}^{(k)}\right),
\label{W_Vl_k_g}
\end{align}
respectively. Since $\mathbf{K}_l$ is not needed in (\ref{R_l_k}) and (\ref{W_Vl_k_g}) (and in $\mathbf{u}_l$, as noted below), updating $\mathbf{K}_l$, which is required in Algorithm~\ref{Alg} through (\ref{K_next}), is also unnecessary.

Now, to demonstrate that the learner's trajectories are not required, consider that by utilizing (\ref{V_NN}) and (\ref{u_e}), we can express $\mathbf{u}_l$ as follows:
\begin{equation}
\mathbf{u}_l= -\frac{1}{2}\mathbf{R}_{l}^{-1}\mathbf{g}^{\top}\nabla\boldsymbol{\sigma}_V^{\top}\mathbf{W}_{Vl}
\label{u_l_g}
\end{equation}
with $\epsilon_V$ neglected. This relation, along with some arbitrary initial values for $\mathbf{R}_l$ and $\mathbf{W}_{Vl}$ (ensuring that $\mathbf{W}_{Vl}^{\top}\boldsymbol{\sigma}_V(\cdot)$ is positive-definite), can be used to set the required values for the first iteration in the estimation of $\mathbf{R}_l$ and $\mathbf{W}_{Vl}$ via (\ref{R_l_k}) and (\ref{W_Vl_k_g}). As a result, the forward RL problem in Step~\ref{Step_initialize_2} of Algorithm~\ref{Alg} becomes unnecessary. Assigning an arbitrary value to $\mathbf{W}_{Vl}$, instead of estimating it through a forward RL using given $\mathbf{R}_l$ and $Q(\cdot)$, is not problematic here because the estimation process of $\mathbf{R}_l$ and $\mathbf{W}_{Vl}$ is performed offline. Thus, any initial or intermediate updated values that might result in an unstable controller will not be applied to the system.

After $\mathbf{R}_{l}$ and $\mathbf{W}_{Vl}$ have converged, which is guaranteed for sufficiently small $\mathcal{E}$, (\ref{HJB_int}) can be used to estimate $Q_l(\cdot)$ using only the trajectories of the expert system. The complete steps of the algorithm are summarized in Algorithm~\ref{Alg2}, and its block diagram is shown in Fig.~\ref{BlockDiagram_Alg2}. Note that Lemma~\ref{Lemma_R} is also applied to Algorithm~\ref{Alg2}.

\begin{algorithm}
\caption{IRL/IOC approach for systems with known input dynamics}\label{Alg2}
\begin{algorithmic}[1]
\item Collect input-state trajectories of the expert;
\item Set the iteration number $k=1$ and assign initial values $\mathbf{R}_{l}^{(k)}\succ0$ and $\mathbf{W}_{Vl}$ such that $\mathbf{W}_{Vl}^{\top}\boldsymbol{\sigma}_V(\cdot)> 0$;
\item Estimate $\mathbf{R}_l$ and $\mathbf{W}_{Vl}$;
\newline
Loop over $k$:
\begin{itemize}
\setlength\itemsep{.0em}
\item Select a sample from the expert's data $(\mathbf{x}_e,\mathbf{u}_e)$;
\item Calculate $\mathbf{u}_l$ at the selected states $\mathbf{x}_e$ using (\ref{u_l_g});
\item Calculate the policy difference: $\mathbf{E}_{u}^{(k)}=\mathbf{u}_{l}-\mathbf{u}_e$;
\item Stop if the recent values of $\mathcal{E}^{(k)}$ over the last several iterations are all less than $\epsilon_E$;
\item Update $\mathbf{R}_{l}$ through (\ref{R_l_k});
\item Update $\mathbf{W}_{Vl}$ through (\ref{W_Vl_k_g});
\end{itemize}
\item Estimate $Q_l(\cdot)$ using (\ref{psi});
\item If $Q_l(\cdot)\leq 0$, go to Step~\ref{Step_initialize_2}.
\end{algorithmic}
\end{algorithm}
\begin{figure}
	\centering
	\includegraphics[scale=.55]{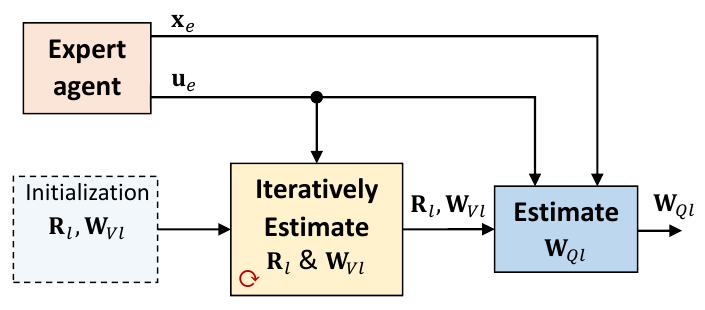}
	\caption{Block diagram of the proposed IRL/IOC method (Algorithm~\ref{Alg2}).}
	\label{BlockDiagram_Alg2}
\end{figure}

%==========================================
%==========================================
\section{Simulation Results}
In this section, we present the simulation results of the proposed IRL/IOC algorithms. Since Algorithm~\ref{Alg2} is similar to Algorithm~\ref{Alg} and is developed specifically for cases where the input dynamics are known, we only present the evaluation results of Algorithm~\ref{Alg}. In all experiments, we assumed that the state and input trajectories of the expert system are available with a sampling rate of $1~\si{m\second}$. The state points for Step~\ref{Step_Update_R_P} of the algorithm were randomly selected at each iteration from states whose amplitudes were larger than a specified threshold, while for Step~\ref{Step_Update_Q}, the points were sampled from the trajectory at intervals of $0.03~\si{\second}$, starting from $t=0~\si{\second}$ until the signal amplitudes fell below a certain threshold. It is noteworthy that only one set of input-state trajectories from the expert agent is necessary, while those from the learner agent are only required during initialization, i.e., Step~\ref{Step_initialize_2} of the algorithm.

\subsection{Results for a single-dimensional input system}
In the first study, we considered the single-dimensional input nonlinear system defined by the following functions \cite{asl2023reinforcement}:
\begin{align}
&\mathbf{f}(\mathbf{x}_e)=\begin{bmatrix}
-x_{e1}+x_{e2} \\
-0.5x_{e1}-0.5x_{e2}\left(1-(\cos(2x_{e1})+2)^2\right)
\end{bmatrix},\nonumber\\
&\mathbf{g}(\mathbf{x}_e)=\begin{bmatrix}
0 \\
\cos(2x_{e1})+2
\end{bmatrix},
\label{Sim_Example1}
\end{align}
with $\mathbf{x}_e=[x_{e1},\;x_{e2}]^{\top}$, where the state vector was initialized at $\mathbf{x}_e(0)=[2,\;2]^{\top}$. For this system, when the cost function is defined with the following terms 
\begin{equation}
Q(\mathbf{x}_e)=x_{e1}^2+x_{e2}^2,\quad R=1,
\end{equation}
the value function can be obtained as $V(\mathbf{x}_e)=0.5x_{e1}^2+x_{e2}^2$ \cite{bhasin2013novel}. We considered $\boldsymbol{\sigma}_V(\mathbf{x}_e)=[x_{e1}^2,\,x_{e1}x_{e2},\,x_{e2}^2]^{\top}$, and $\boldsymbol{\sigma}_g(\mathbf{x}_e)=[1,\;0;\; 0,\; 1;\; 0,\;\cos(2x_{e1})]\in\mathbb{R}^{3\times 2}$. Hence, we obtained $\mathbf{K}_e=[0,\, 1,\, 0,\, 0.5]$, which was used only for comparison with the estimated control gain $\mathbf{K}_l$.

To estimate the parameters, we initially defined the cost function terms as $Q_l({\cdot})=0.5x_{e1}^2+1.5x_{e2}^2$ and $R_l=0.8$. We used the approach in \cite{lee2014integral} to estimate $\mathbf{W}_{Vl}$ and $\mathbf{K}_l$ corresponding to the given cost function, obtaining $\mathbf{W}_{Vl}=[0.2664,\, -0.0526,\, 0.8623]^{\top}$ and $\mathbf{K}_l=[-0.0657,\,1.0779,\,-0.0328,\,0.5389]$. Using the initial value of $\mathbf{W}_{Vl}$ and considering $\boldsymbol{\sigma}_V$ and $\boldsymbol{\sigma}_g$, we then obtained the initial value for $\mathbf{W}_{ul}$ as
\begin{equation}
\mathbf{W}_{ul}=\begin{bmatrix}
0.2664 & -0.0526 & 0 & 0\\
-0.0526 & 0.8623 & 0 & 0\\
0 & 0 & -0.0526 & 0.8623
\end{bmatrix}.
\end{equation}
% $\mathbf{W}_{ul}=[0.2664,\;-0.0526,\;0,\;0;\;-0.0526,\;0.8623,\;0,\;0;\;0,\;0,\;-0.0526,\;0.8623]$
Since the considered example is a single-dimensional input system, according to Lemma~\ref{Lemma_R}, we did not update $R_l$ in Step~\ref{Step_Update_R_P} of the algorithm. The update gain for estimating $\mathbf{W}_{Vl}$ was set to $\alpha_V=0.003$. The convergence of $\mathbf{W}_{Vl}$ to an equivalent value, and $\mathbf{K}_l$ to $\mathbf{K}_e$ is shown in Fig.~\ref{K_L_Sim1}.

The estimated $\mathbf{W}_{Vl}$, along with the initial value of $R_l$, was then used to estimate $\mathbf{W}_{Ql}$, with $\boldsymbol{\sigma}_Q=[x_{e1}^2,\,2x_{e1}x_{e2},\,x_{e2}^2]^{\top}$, yielding
\begin{equation}
\mathbf{W}_{Ql}=[0.5252,\;0.1368,\;0.8026]^{\top}.
\end{equation}
To verify the estimated parameters, we applied them in the forward RL algorithm from \cite{lee2014integral}, resulting in an estimated control gain vector of $[-0.0003,\;1.0003,\;-0.0001,\;0.5002]$, which is very close to the true control gain $\mathbf{K}_e$. 
\begin{figure}
	\centering
	\includegraphics[scale=.44]{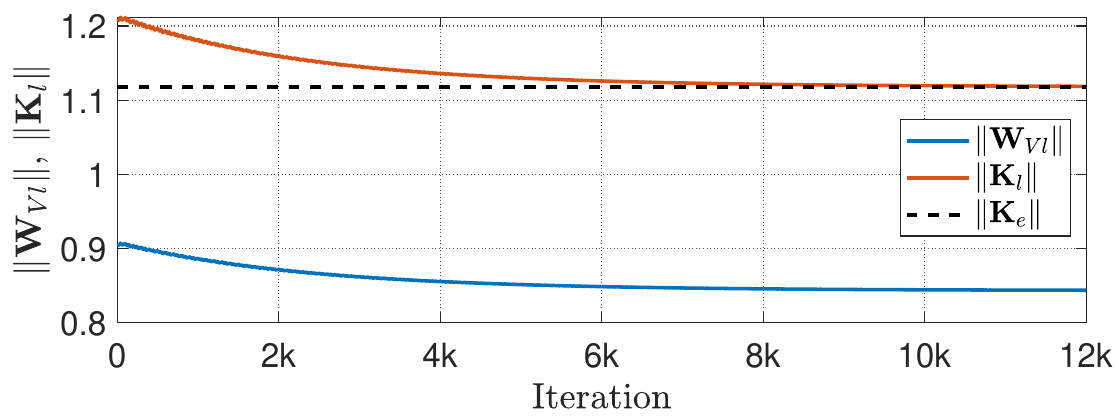}
	\caption{Evolution of the norm of the control gain $\mathbf{K}_l$ and the parameter vector $\mathbf{W}_{Vl}$ for the system (\ref{Sim_Example1}).}
	\label{K_L_Sim1}
\end{figure}

\subsection{Results for a quadrotor}
In this study, we considered the rotational dynamics of a quadrotor aerial vehicle. We tested the method across two simulation frameworks: MATLAB and a Python-based implementation using the MuJoCo physics engine. In the MATLAB environment, we isolated the rotational dynamics, modeling the quadrotor's attitude in a controlled setting without additional subsystems, enabling a focused analysis of the IOC solution. In contrast, the Python-MuJoCo environment offered a physically realistic simulation of the full quadrotor model, where we addressed the rotational dynamics while employing a PID controller to stabilize altitude and maintain a hover state---a practical necessity to facilitate testing of the rotational control.

In MATLAB, we considered the rotational dynamics of a quadrotor aerial vehicle, defined by the following functions:
\begin{align}
&\mathbf{f}(\cdot)=
\begin{bmatrix}
\dot{\phi}\\
\dot{\theta}\\
\dot{\psi}\\
-\frac{I_{yy}-I_{zz}}{I_{xx}}\dot{\theta}\dot{\psi}\\
-\frac{I_{zz}-I_{xx}}{I_{yy}}\dot{\phi}\dot{\psi}\\
-\frac{I_{xx}-I_{yy}}{I_{zz}}\dot{\phi}\dot{\theta}
\end{bmatrix},\,\mathbf{g}=\begin{bmatrix}
\mathbf{0}_{3\times 3}\\
%\begin{bmatrix}
%x_{e1}\\
%x_{e2}\\
%x_{e3}\\
%-\frac{I_{yy}-I_{zz}}{I_{xx}}x_{e2}x_{e3}\\
%-\frac{I_{zz}-I_{xx}}{I_{yy}}x_{e1}x_{e3}\\
%-\frac{I_{xx}-I_{yy}}{I_{zz}}x_{e1}x_{e2}
%\end{bmatrix},\,\mathbf{g}=\begin{bmatrix}
%\mathbf{0}_{3\times 3}\\
\text{diag}([\frac{1}{I_{xx}},\frac{1}{I_{yy}},\frac{1}{I_{zz}}])
\end{bmatrix},
\label{Sim_Example2}
\end{align}
where $I_{xx}=4.856\times 10^{-3}$, $I_{yy}=4.856\times 10^{-3}$, and $I_{zz}=8.801\times 10^{-3}$ are the moments of inertia. The state vector is defined as $\mathbf{x}_e=[\phi,\,\theta,\,\psi,\,\dot{\phi},\,\dot{\theta},\,\dot{\psi}]^{\top}$, where $\phi$, $\theta$, $\psi$ represent the the roll, pitch, and yaw orientation of the vehicle. The state vector was initialized at $\mathbf{x}_e(0)=[1.5,\,1.7,\,1.8,\,0,\,0,\,0]^{\top}$, and the expert's cost function was defined by $\mathbf{R}=\text{diag}([1,\;1,\;1])$ and $Q(\mathbf{x}_e)=\mathbf{x}_e^{\top}\bar{\mathbf{Q}}\mathbf{x}_e$, where the matrix $\bar{\mathbf{Q}}$ is given by
\begin{align}
\bar{\mathbf{Q}}=\begin{bmatrix}
3 & 0 & 0 & 0.2 & 0 & 0\\
0 & 3 & 0 & 0 & 0.2 & 0\\
0 & 0 & 3 & 0 & 0 & 0.2\\
0.2 & 0 & 0 & 1 & 0 & 0\\
0 & 0.2 & 0 & 0 & 1 & 0\\
0 & 0 & 0.2 & 0 & 0 & 1
\end{bmatrix}.
\end{align}
Using a forward RL method and considering $\boldsymbol{\sigma}_V(\mathbf{x}_e)=[\phi^2,\,\theta^2,\,\psi^2,\,\dot{\phi}^2,\,\dot{\theta}^2,\,\dot{\psi}^2,\,\phi\dot{\phi},\,\theta\dot{\theta},\,\psi\dot{\psi}]^{\top}$, along with the fact that $\mathbf{g}$ is a constant matrix, we obtained the optimal gain $\mathbf{K}_e$ for the given cost function as follows:
\begin{equation}
\mathbf{K}_e=\begin{bmatrix}
1.732 & 0 & 0 & 1 & 0 & 0\\
0 & 1.732 & 0 & 0 & 1 & 0\\
0 & 0 & 1.732 & 0 & 0 & 1
\end{bmatrix}.
\label{K_e_Sim2}
\end{equation}

For the estimation, the initial cost terms were set as $\mathbf{R}_l=\text{diag}([0.8,\,0.5,\,0.7])$ and $Q_l(\cdot)=\mathbf{x}_l^{\top}\bar{\mathbf{Q}}_l\mathbf{x}_l$, with $\bar{\mathbf{Q}}_l=\text{diag}([1,\,1,\,1,\,1,\,1,\,1])$. We used an RL approach to estimate $\mathbf{W}_{Vl}$ and $\mathbf{K}_l$ corresponding to the initial cost function, obtaining $\mathbf{W}_{Vl}=[1.0043,\,1.0034,\,1.0074,\allowbreak\,0.0044,\,0.0034,\,0.0074,\,0.0087,\,0.0069,\,0.0147]^{\top}$. Using the initial value of $\mathbf{W}_{Vl}$ and considering $\boldsymbol{\sigma}_V$, we then obtained the initial value for $\mathbf{W}_{ul}$ as
\begin{align}
\mathbf{W}_{ul}=\begin{bmatrix}
2.0086 & 0 & 0 & .0087 & 0 & 0\\
0 & 2.0068 & 0 & 0 & .0069 & 0\\
0 & 0 & 2.0068 & 0 & 0 & .0147\\
.0087 & 0 & 0 & .0088 & 0 & 0\\
0 & .0069 & 0 & 0 & .0068 & 0\\
0 & 0 & .0147 & 0 & 0 & .0148
\end{bmatrix}.
\end{align}
The update gains were set as $\alpha_V=\alpha_R=10^{-6}$. The convergence of $\mathbf{W}_{Vl}$ to an equivalent value, and $\mathbf{K}_l$ to $\mathbf{K}_e$ is shown in Fig.~\ref{K_L_Sim2}.

The estimated parameters were used to calculate $\mathbf{W}_{Ql}$ by considering $\boldsymbol{\sigma}_Q=[\phi^2,\,2\phi\theta,\,2\phi\psi,\,\cdots,\,\theta^2,\,2\theta\psi,\cdots,\,\dot{\psi}^2]^{\top}\in\mathbb{R}^{21}$, which includes all distinct quadratic terms of the components of $\mathbf{x}_e$, incorporating both squared terms and cross-products. We obtained
\begin{align}
\mathbf{W}_{Ql}=
&[2.3829,\;0.0009,\;0.0010,\;0.3716,\;0.0006,\;0.0007,\nonumber\\
&\;1.4961,\;0.0009,\;-0.0009,\;-0.1405,\;0.0015,\nonumber\\
&\;2.0949,\;-0.0006,\;0.0008,\;0.2021,\;0.7817,\nonumber\\
&\;0.0003,\;0.0007,\;0.4892,\;-0.0004,\;0.6768]^{\top}.
\end{align}
To verify the estimated parameters, we applied them in a forward RL algorithm, and the estimated control gain, denoted by $\hat{\mathbf{K}}_l$, is as follows:
\begin{align}
&\hat{\mathbf{K}}_l=\nonumber\\
&\begin{bmatrix}
1.7262 & -.0007 & .0002 & .9971 & -.0001 & .0006\\
0.0028 & 1.7303 & .0035 & .0012 & .9976 & .0005\\
0.0010 & .0008  & 1.7303 & .0005 & -.0001 & .9991
\end{bmatrix},
\end{align}
which is very close to the true control gain (\ref{K_e_Sim2}). 
%Fig.~\ref{K_L_Sim2} shows the evolution of the control $\mathbf{K}_l$, updated in Step~\ref{Step_Update_R_P} of Algorithm~\ref{Alg}.
\begin{figure}
	\centering
	\includegraphics[scale=.44]{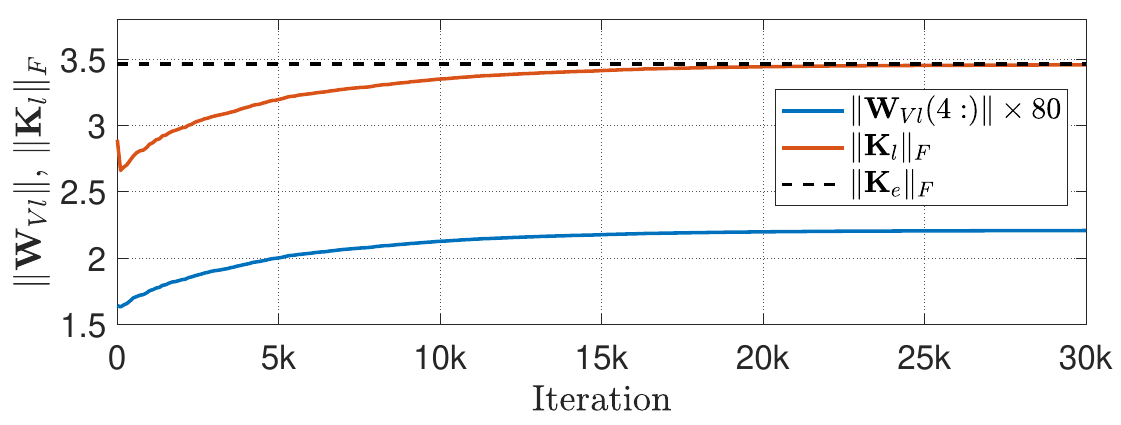}
	\caption{Evolution of the norm of the control gain $\mathbf{K}_l$ and the parameter vector $\mathbf{W}_{Vl}$ for the system (\ref{Sim_Example2}).}
	\label{K_L_Sim2}
\end{figure}

For the MuJoCo simulations, we utilized a quadrotor model from the \texttt{gym\_rotor} repository \cite{gym_rotor}, which defines the physical properties and dynamics of the aerial vehicle.

The state vector was initialized at $\mathbf{x}_e(0)=[0.287,\,0.268,\,0.427,\,0,\,0,\,0]^{\top}$, and the expert's cost function was defined the same as the previous simulation. We solve the forward problem for this system and the approximate control gain obtained as follows:
\begin{equation}
\mathbf{K}_e=\begin{bmatrix}
1.732 & -.005 & -.021 & .434 & -0.006 & -0.016\\
0.007 & 1.740 & .038 & .004 & .435 & .001\\
.021 & -.026 & 1.725 & .016 & -.001 & .453
\end{bmatrix}.
\end{equation}
For the learner system, we assigned the initial cost function the same as the previous simulation and estimated the initial $\mathbf{W}_{Vl}$ as $\mathbf{W}_{Vl}=[1.3572,\,1.26271,\,1.3079,\allowbreak\,0.0069,\,0.0068,\,0.0141,\,0.0277,\,0.0222,\,0.0546]^{\top}$. After the convergence of $\mathbf{W}_{Vl}$ and $\mathbf{K}_{l}$, we estimated $\mathbf{W}_{Ql}$ as follows:
\begin{align}
\mathbf{W}_{Ql}=
&[2.4589,\;.0264,\;-.0318,\;.6202,\;.0120,\;.0277,\nonumber\\
&\;1.3965,\;.1488,\;-.0113,\;-.1378,\;-.1153,\nonumber\\
&\;2.0949,\;.0016,\;-.0926,\;.4277,\;.7244,\nonumber\\
&\;.0010,\;-.0109,\;.4257,\;-.0012,\;.6971]^{\top}.
\end{align}
We utilized the estimated cost function to obtain the optimal policy using a model-free RL method. The evolution of the norm of the control policy and its convergence to the optimal one is shown in Fig.~\ref{Mujoco_Photo}.
\begin{figure}
	\centering
	\includegraphics[scale=.52]{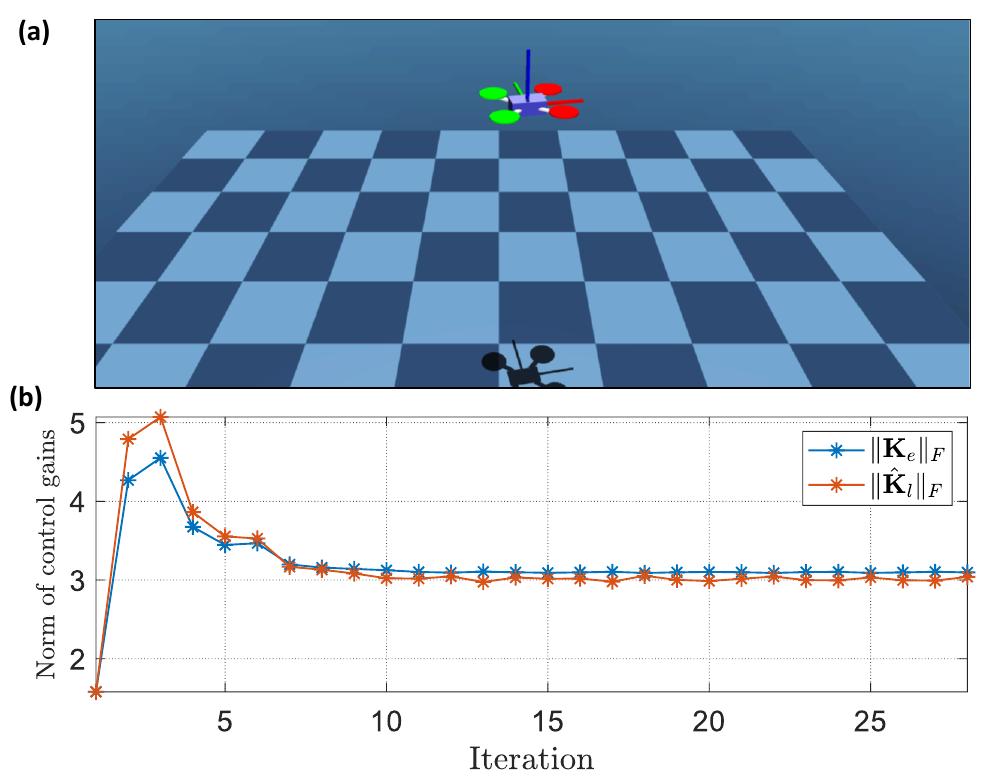}
	\caption{(a) MuJoCo environment used for the quadrotor simulation. (b) Evolution of the norm of control policies estimated through model-free RL using the expert's original cost function and the estimated cost function.}
	\label{Mujoco_Photo}
\end{figure}

\subsection{Noise and uncertainty analysis}
In this study, we evaluated the performance of Algorithm~\ref{Alg} in the presence of measurement noise and errors in the initial estimates of $\mathbf{W}_{Vl}$ and $\mathbf{K}_l$. We also compared the results with the IRL method presented in \cite{SELF2022110242} under the same conditions. To analyze the robustness of the algorithms under these conditions, we measured the difference between the true optimal policy and the policy obtained from the estimated cost function. To reduce the impact of potential errors in finding the optimal policy from a given cost function, an issue common in nonlinear systems, for a better analysis and comparison of robustness, we considered a linear system with dynamics defined by the following terms:
\begin{equation}
\mathbf{f}=
\begin{bmatrix}
1 & 0\\
0 & -2
\end{bmatrix}\mathbf{x}_e,\quad
\mathbf{g}=\begin{bmatrix}
1\\
1
\end{bmatrix}.
\label{Sim_Example3}
\end{equation}
The state vector was initialized at $\mathbf{x}_e(0)=[-0.5,\;0.5]^{\top}$. The expert's cost function was defined with $R=1$ and $Q(\mathbf{x}_e)=\mathbf{x}_e^{\top}\bar{\mathbf{Q}}\mathbf{x}_e$, where $\bar{\mathbf{Q}}=\text{diag}([1,\;0.6])\in\mathbb{R}^{2\times 2}$. The initial cost weights were set as $R_l=0.8$ and $\mathbf{W}_{Ql}=[0.5,\;0,\;1.3]^{\top}$ with $\boldsymbol{\sigma}_Q=[x_{e1}^2,\;2x_{e1}x_{e2},\;x_{e2}^2]^{\top}$. For this system, the optimal control gain was obtained as $\mathbf{K}_e=[2.4877,\;0.0429]^{\top}$.

To simulate measurement noise, uniform random noise was added to the states and input of the expert, with the noise amplitude limited to $\pm 3\%$ of the signal values at each time step. We used MATLAB's \texttt{rand} function to generate a separate random number for each signal and applied these signals for estimation without using any filter. To model uncertainty in the initial estimates of $\mathbf{W}_{Vl}$ and $\mathbf{K}_l$, we introduced uncertainty in the input dynamics of the system $\mathbf{g}$ when solving the forward problem in Step~\ref{Step_initialize_2} of the algorithm to obtain the initial parameters. This approach to handling uncertainty provided a fair framework for comparing our algorithm with the model-based approach in \cite{SELF2022110242}. In simulating the method from \cite{SELF2022110242}, we assumed that the drift dynamics $\mathbf{f}$ were correctly known and used the same noisy signals as in our method.

After estimating the cost function, we solved the forward problem to calculate the error between the true control gain and the one obtained from the estimated cost function. The normalized policy errors obtained using Algorithm~\ref{Alg} and the method from \cite{SELF2022110242} for different levels of uncertainty in $\mathbf{g}$ are shown in Fig.~\ref{Noise_error}. The uncertainty in $\mathbf{g}$ was introduced by increasing its true value by up to $10\%$. Note that the noise signals were regenerated for each new uncertainty in $\mathbf{g}$. The results in Fig.~\ref{Noise_error} demonstrate that both methods exhibit comparable sensitivity to noise and uncertainty, while the model-based approach, as noted in \cite{SELF2022110242} and \cite{ASL_ESWA2024}, is limited to IRL problems that admit unique solutions up to a scaling factor.
 \begin{figure}
	\centering
	\includegraphics[scale=.44]{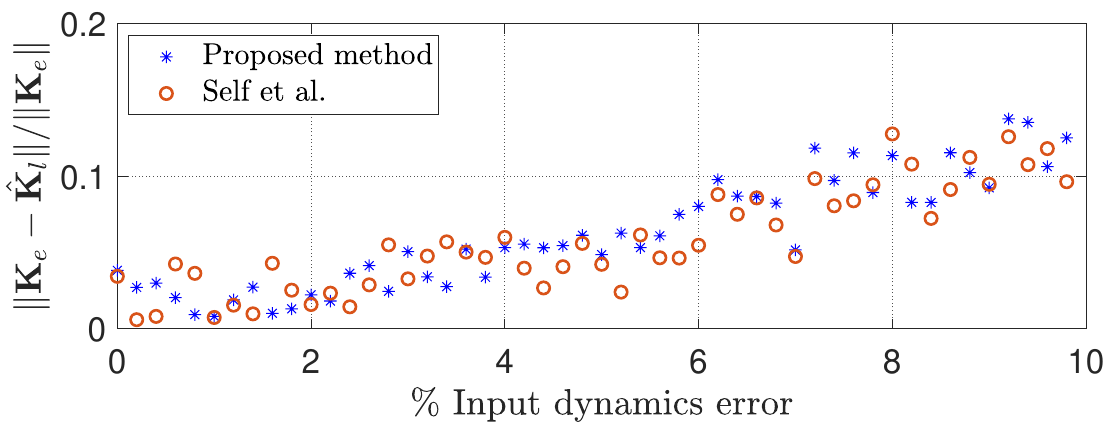}
	\caption{Comparison of normalized policy errors between our proposed method and the approach in \cite{SELF2022110242} for different levels of uncertainty in the input dynamics $\mathbf{g}$, using different noisy signals in each experiment. $\hat{\mathbf{K}}_l$ represents the policy gain obtained from the estimated cost function.}
	\label{Noise_error}
\end{figure}

%==========================================
%==========================================
\section{Conclusion}
This paper introduced novel solutions for addressing inverse optimal control in continuous-time nonlinear systems. By employing a dual approach that combines gradient descent techniques for estimating one set of parameters and the Hamilton-Jacobi-Bellman equation for estimating another, our algorithms effectively overcome the issue of the finite informativity condition, which may not be satisfied for certain systems in conventional approaches. Additionally, they reduce complexity by eliminating the need for repetitively solving a forward optimal control problem. In fact, the proposed model-free algorithm requires solving the forward problem only once during initialization, and we showed that for systems with known input dynamics, this step can be bypassed entirely. Through extensive simulations under both noise-free and noisy conditions, we demonstrated the practicality of our approach in estimating the cost function of nonlinear systems. The demonstrated effectiveness suggests promising prospects for deployment in real-world scenarios, particularly in the realm of autonomous systems and robotics. Future research could focus on refining the approach's applicability across a wider range of practical contexts. 

%========================================
%========================================
\section*{Acknowledgments}
This study is based on results obtained from a project (JPNP20006) commissioned by the New Energy and Industrial Technology Development Organization (NEDO).
%==========================================
%==========================================

\bibliographystyle{IEEEtran}
\bibliography{bib_IOC}

%==========================================
%==========================================

\end{document}